\documentclass{article}
\usepackage{ijcai24_cr}

\usepackage{lmodern}
\usepackage{times}
\usepackage{soul}
\usepackage{url}
\usepackage[hidelinks]{hyperref}
\usepackage[utf8]{inputenc}
\usepackage[small]{caption}
\usepackage{graphicx}
\usepackage{amsmath}
\usepackage{amsthm}
\usepackage{amssymb}
\usepackage{booktabs}
\usepackage{algorithm}
\usepackage{algorithmic}
\usepackage[switch]{lineno}
\usepackage[inline]{enumitem}
\usepackage{mathtools}
\usepackage{pgfplots}
\usepackage{subcaption}
\usepackage{thm-restate}
\usepackage{xcolor} % color for editing
\usepackage{xfrac} % for \sfrac
\usepackage{xspace}

\usepackage[greek,english]{babel}

\newcommand{\pki}[1]{#1} % 
\newcommand{\Petros}[1]{#1}

\definecolor{darkgreen}{rgb}{0,0.5,0}

\newcommand{\er}{Erd\H{o}s-R\'{e}nyi\xspace}
\newcommand{\scfr}{Scale-Free\xspace}
\newcommand{\pr}{\mathbf{PR}}
\newcommand{\transit}{\mathbf{T}}
\newcommand{\States}{\mathcal{S}}
\newcommand{\RewardStates}{\ensuremath{\mathcal{S}_{\mathcal{R}}}\xspace}
\newcommand{\RewardStatesA}{\ensuremath{\mathcal{A}}\xspace}
\newcommand{\RewardStatesB}{\ensuremath{\mathcal{B}}\xspace}
\newcommand{\RewardStatesX}{\ensuremath{\mathcal{X}}\xspace}
\newcommand{\RewardStatesAcapB}{\ensuremath{\mathcal{A} \cap \mathcal{B}}\xspace}
\newcommand{\RewardStatesAcupB}{\ensuremath{\mathcal{A} \cup \mathcal{B}}\xspace}
\newcommand{\RewardStatesone}{\ensuremath{\mathcal{S}_{R_{1}}}\xspace}
\newcommand{\RewardStatestwo}{\ensuremath{\mathcal{S}_{R_{2}}}\xspace}
\newcommand{\RewardStatesPolicy}{\ensuremath{\mathcal{S}^*_\policy}\xspace}
\newcommand{\RewardStatesPolicyi}{\ensuremath{\mathcal{S}^*_{\policy_i}}\xspace}
\newcommand{\state}{s}

\newcommand{\Transitions}{\mathcal{T}}
\newcommand{\RewardVector}{\mathcal{R}}

\newcommand{\Initial}{{\mathcal{I}}}
\newcommand{\Moves}{{\mathcal{M}}}
\newcommand{\currentreward}{r}
\newcommand{\policy}{\pi}
\newcommand{\Policies}{\Pi}

\newcommand{\ExpReward}{{F}}

\newtheorem{theorem}{Theorem}
\newtheorem{corollary}[theorem]{Corollary}
\newtheorem{lemma}[theorem]{Lemma}

\DeclareMathOperator*{\argmax}{arg\,max}

\newcommand{\maxmin}{$\max$--$\min$\xspace}
\newcommand{\minmax}{$\min$--$\max$\xspace}
\newcommand{\zok}{\textsc{0--1 Knapsack}\xspace}
\newcommand{\mnk}{\textsc{Max--Min 0--1 Knapsack}\xspace}
\newcommand{\hitset}{\textsc{Hitting Set}\xspace}
\newcommand{\mstree}{\textsc{Minimum Steiner Tree}\xspace}
\newcommand{\knapsack}{\textsc{Knapsack}\xspace}
\newcommand{\problemnorobust}{Reward Placement\xspace}
\newcommand{\problemnorobustshort}{RP\xspace}
\newcommand{\problem}{Robust Reward Placement\xspace}
\newcommand{\problemshort}{\textsc{RRP}\xspace}
\newcommand{\setting}{setting\xspace}
\newcommand{\settings}{settings\xspace}
\newcommand{\modelname}{Markov Mobility Model\xspace}
\newcommand{\modelnames}{Markov Mobility Models\xspace}
\newcommand{\modelshorts}{MMMs\xspace}
\newcommand{\modelshort}{MMM\xspace}

\newcommand*{\NPhard}{\ensuremath{\mathbf{NP}}-hard\xspace}

\title{Robust Reward Placement under Uncertainty}
\author{
Petros Petsinis$^1$\and
Kaichen Zhang$^2$\and
Andreas Pavlogiannis$^1$\and
Jingbo Zhou$^3$\And\\
Panagiotis Karras$^{4,1}$
\affiliations
$^1$Department of Computer Science, Aarhus University\\
$^2$Artificial Intelligence Thrust, Hong Kong University of Science and Technology (Guangzhou)\\
$^3$Business Intelligence Lab, Baidu Research\\
$^4$Department of Computer Science, University of Copenhagen
\emails
petsinis@cs.au.dk\and
kzhangbi@connect.ust.hk\and
pavlogiannis@cs.au.dk\and
zhoujingbo@baidu.com\and\\
piekarras@gmail.com
}

\begin{document}
\maketitle
%Abstract
%-------------------------------------------------------------
\begin{abstract}
We consider a problem of placing generators of rewards to be collected by randomly moving agents in a network. In many settings, the precise mobility pattern may be one of several possible, based on parameters outside our control, such as weather conditions. The placement should be \emph{robust} to this uncertainty, to gain a competent total reward across possible networks. To study such scenarios, we introduce the \emph{Robust Reward Placement} problem (RRP). Agents move randomly by a Markovian Mobility Model with a predetermined set of locations whose connectivity is chosen adversarially from a known set~$\Pi$ of candidates. We aim to select a set of reward states within a budget that maximizes the minimum ratio, among all candidates in~$\Pi$, of the collected total reward over the optimal collectable reward under the same candidate. We prove that RRP is NP-hard and inapproximable, and develop~$\Psi$-Saturate, a pseudo-polynomial time algorithm that achieves an~$\epsilon$-additive approximation by exceeding the budget constraint by a factor that scales as~$\mathcal{O}(\ln |\Pi|/\epsilon)$. In addition, we present several heuristics, most prominently one inspired by a dynamic programming algorithm for the \maxmin \zok problem. We corroborate our theoretical analysis with an experimental evaluation on synthetic and real data.
\end{abstract}
%-------------------------------------------------------------

%Introduction
%-------------------------------------------------------------
\section{Introduction}

In many graph optimization problems, a stakeholder has to select locations in a network, such as a road, transportation, infrastructure, communication, or web network, where to place reward-generating facilities such as stores, ads, sensors, or utilities to best service a population of moving agents such as customers, autonomous vehicles, or bots~\cite{zhang2016submodular,ostachowicz2019optimization,zhang2020geodemographic,rosenfeld2016optimal,amelkin2019fighting}. 
Such problems are intricate due to the uncertainty surrounding agent mobility~\cite{krause2008robust,chen2016robust,rimkempe,horvcik2022optimal}.

\begin{figure}[!t]
\hspace*{0mm}
\includegraphics[width=0.45\textwidth]{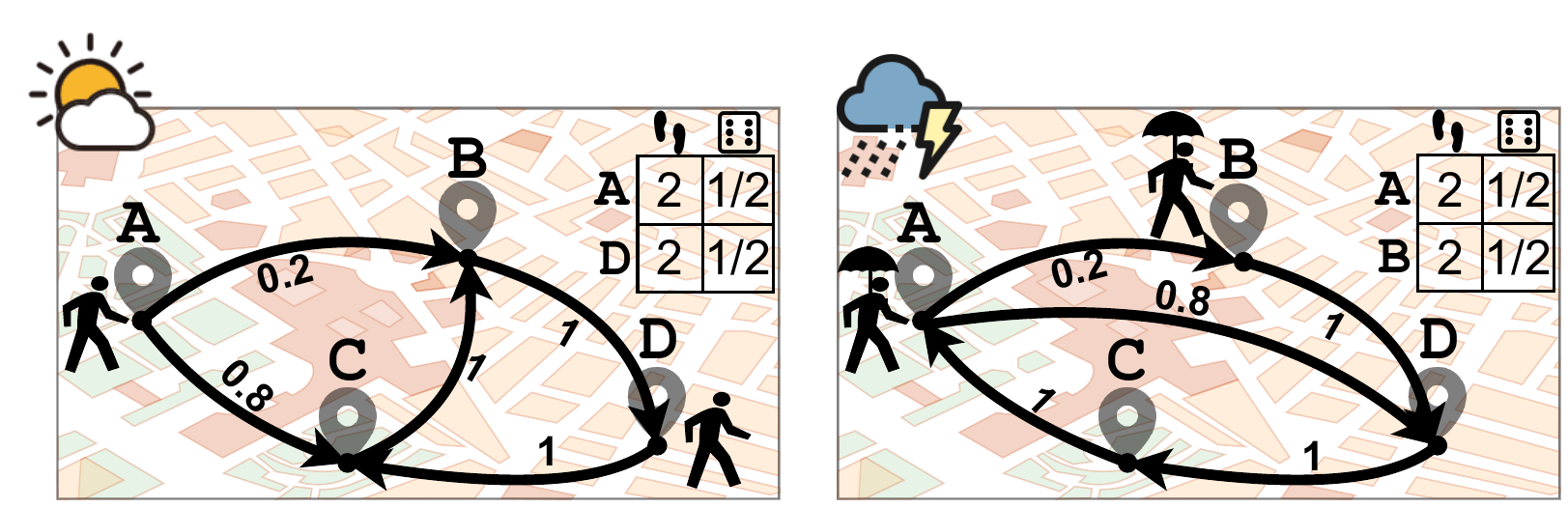}
\caption{Moving agent under two settings; \emph{sunny} and \emph{rainy}; tables show numbers of steps and initial probabilities.}\label{fig:ad}
\end{figure}

For instance, consider outdoor ad placement. We represent the road map as a probabilistic network in which agents move. If every agent follows the same movement pattern regardless of environmental conditions, then the problem of placing ads to maximize the expected number of ad views admits a greedy algorithm with an approximation ratio~\cite{zhang2020geodemographic}. Still, the problem becomes more \pki{involved} under malleable environmental conditions that alter movement patterns. As a toy example, Figure~\ref{fig:ad} shows a probabilistic network. An agent randomly starts from an initial location and takes two steps by the probabilities shown on edges representing street segments, under two environmental settings, \emph{sunny} and \emph{rainy}. Assume a stakeholder has a budget to place an ad-billboard at a \emph{single location}. Under the \emph{sunny} setting, the best choice of placement is~$B$, as the agent certainly passes by that point regardless of its starting position; under the \emph{rainy} setting, the agent necessarily passes by~$D$ within two steps, hence that is most preferable. However, under the \emph{rainy} setting~$B$ yields expected reward~$0.6$, and so does~$D$ under the \emph{sunny} one. Due to such uncertainty, a risk-averse stakeholder would prefer the location that yields, in the worst case, the highest ratio of the collected to best feasible reward, i.e., in this case, $C$, which yields expected reward~$0.9$ under both settings.

In this paper, we introduce the problem of robust reward placement (\problemshort) in a network, under uncertainty about the environment whereby an agent is moving according to any of several probabilistic mobility settings. We express each such setting by a \modelname (\modelshort)~$\policy~\in~\Policies$. The cumulative reward a stakeholder receives grows whenever the agent passes by one of the \emph{reward states}~$\RewardStates$. \problemshort seeks to select a set of such states~$\RewardStates^*$ within a budget, that maximizes the worst-case ratio, across all settings~$\Policies$, of the collected reward~$\ExpReward(\RewardStates|\policy)$ over the highest reward that can be collected under the same setting $\ExpReward(\RewardStatesPolicy|\policy)$, i.e., $\RewardStates^*~=~\arg\max_{\RewardStates} \min_{\policy \in \Policies} \frac{\ExpReward(\RewardStates|\policy)}{\ExpReward(\RewardStatesPolicy|\policy)}$. This max-min ratio objective is used in risk-averse portfolio optimization and advertising~\cite{ordentlich1998cost,li2020optimal}.

\paragraph{Our Contribution.} 
Our contributions stand as follows:
\begin{enumerate}[noitemsep,topsep=0pt,parsep=0pt,partopsep=0pt]
\item 
We introduce the problem of~\problem~(\problemshort) over a set of~\modelnames, that has real-world applications across various domains.

\item We study the properties of \problemshort and show that it is \NPhard (Theorem~\ref{th:rp_np_hard_Pi_1}). Due to the additivity and monotonicity properties of the reward function (Lemma~\ref{lem:monot_addit}), it admits an optimal solution in pseudo-polynomial time under a single \setting, i.e. $|\Policies| = 1$ (Lemma~\ref{lem:rp_optimal_linear}), yet it is inapproximable when $|\Policies|>1$ unless we exceed the budget constraint by a factor $\mathcal{O}(\ln |\Pi|)$ (Theorem~\ref{th:inapprox}). 

\item We adopt techniques from robust influence maximization to develop~$\Psi$-Saturate, a pseudo-polynomial time algorithm that finds a solution within~$\epsilon$ distance of the optimal, i.e. $\text{OPT}-\epsilon$, while exceeding the budget constraint by an~$\mathcal{O}(\ln\sfrac{|\Policies|}{\epsilon})$ factor (Lemma~\ref{lem:bicriteria_apx}).

\item We present several heuristics as alternative solutions, most prominently one based on a dynamic programming algorithm for the \maxmin \zok problem, to which RRP can be reduced (Lemma~\ref{lem:rp_and_minmax}).
\end{enumerate}

We corroborate our analysis with an experimental evaluation on synthetic and real data. Due to space constraints, we relegate some proofs to the Appendix~\ref{sec:app}.
%-------------------------------------------------------------

%Related Work
%-------------------------------------------------------------
\section{Related Work}

The \problem{} problem relates to robust maximization of spread in a network, with some distinctive characteristics. Some works~\cite{du2013scalable,rimkempe,chen2016robust,logins20,logins22} study problems of selecting a seed set of nodes that robustly maximize the expected spread of a diffusion process over a network. However, in those models~\cite{kempe2003maximizing} the diffusion process is \emph{generative}, whereby an item propagates in the network by producing unlimited replicas of itself. On the other hand, we study a \emph{non-generative} spread function, whereby the goal is to reach as many as possible out of a population of network-resident agents. Our spread function is similar to the one studied in the problem of Geodemographic Influence Maximization~\cite{zhang2020geodemographic}, yet thereby the goal is to select a set of network locations that achieves high spread over a mobile population under a single environmental setting. We study the more challenging problem of achieving competitive spread in the worst case under uncertainty regarding the environment.

Several robust discrete optimization problems~\cite{kouvelis2013robust} address uncertainty in decision-making by optimizing a~\maxmin or~\minmax function under constraints.
The robust \mstree problem~\cite{johnson2000prize} seeks to minimize the worst-case cost of a tree that spans a graph; the~\minmax and~\minmax regret versions of the \knapsack problem~\cite{aissi2009min} have a modular function as a budget constraint; other works examine the robust version of submodular functions~\cite{submodular_survey12,rimkempe} that describe several diffusion processes \cite{adiga2014sensitivity,krause2008robust}. 
To our knowledge, no prior work considers the objective of maximizing the worst-case ratio of an additive function over its optimal value subject to a knapsack budget constraint.
%-------------------------------------------------------------

%Preliminaries
%-------------------------------------------------------------
\section{Preliminaries} 

\paragraph{\modelname~(\modelshort).} 
We denote a discrete-time \modelshort{} as~$\policy = (\States, \Initial, \Transitions, \Moves)$, where~$\States$ is a set of~$n$ states, $\Initial$ is a vector of~$n$ elements in~$[0,1]$ expressing an initial probability distribution over states in~$\States$, $\Transitions$ is an $n \times n$ right-stochastic matrix, where~$\Transitions[\state, \state']$ is the probability of transition from state~$\state \in \States$ to another state~$\state' \in \States$, and~$\Moves$ is an~$n \times K$ matrix with elements in~$[0,1]$, where~$K$ is the maximum number of steps and~$\Moves[\state, k]$ expresses the cumulative probability that an agent starting from state~$\state\in\States$ takes~$k' \in [k, K]$ steps.
Remarkably, an \modelshort describes multiple agents and movements, whose starting positions are expressed via initial distribution $\Initial$ and their step-sizes via $\Moves$.

\paragraph{Rewards.} 
Given an \modelshort, we select a set of states to be \emph{reward states}. 
We use \emph{a reward vector}~$\RewardVector\in\{0, 1\}^n$ to indicate whether state~$\state \in \States$ is a \emph{reward state} and denote the \emph{set} of reward states as~$\RewardStates = \{ \state \in \States | \RewardVector[s] = 1 \}$. In each timestamp~$t$, an agent at state~$\state$ 
%may move 
moves
to state~$\state'$ and 
%retrieve 
retrieves
reward~$\RewardVector[\state']$.
For a set of reward states~\RewardStates,  
%with reward vector~$\RewardVector$, 
and a given \modelshort~$\policy$, the \emph{cumulative reward}~$\ExpReward(\RewardStates|\policy)$ of an agent equals:
\begin{align}
\ExpReward(\RewardStates|\policy) &= \sum_{k\in[K]} \ExpReward_{\policy}(\RewardStates|k)\label{eq:expreturn}\\
\ExpReward_{\policy}(\RewardStates|k) &= \pki{\RewardVector^{\top}} \left( \Transitions^k (\Initial \circ \Moves_k) \right),\label{eq:expreturn_detail}
\end{align}
where $\ExpReward_{\policy}(\RewardStates|k)$ is the expected reward at the $k^{\text{th}}$ step, $\Moves_k$ is the~$k^{\text{th}}$ column of~$\Moves$, and $\circ$ denotes the Hadamard product. 
%Note that as $K\rightarrow \infty$, Equation~\ref{eq:expreturn_detail} yields the steady-state distribution of the model.
%Equation~\ref{eq:expreturn_detail} is a general formulation of PageRank scores \cite{brin1998anatomy} as it considers different initial and step distributions via $\Initial$ and $\Moves$, respectively.

\paragraph{Connection to Pagerank.} The Pagerank algorithm~\cite{brin1998anatomy}, widely used in recommendation systems, computes the stationary probability distribution of a random walker 
%navigating 
in a network. 
%While Pagerank scores can be computed in closed form by eigenvector analysis, that is prohibitive in terms of time efficiency. 
%The power-iteration method~\cite{mises1929praktische} accelerates the computation. 
The Pagerank scores are efficiently computed via power-iteration method~\cite{mises1929praktische}. 
Let~$\pr$ be an~$N \times 1$ column-vector of the Pagerank probability scores, initialized as~$\pr(0)$, 
%and computed across iterations, 
$\transit$ is an~$N \times N$ matrix featuring the transition probabilities 
%that the walker moves among nodes. 
of walker,  
and $\mathbf{1}$ be the all-ones vector.
For a damping factor~$a$, the power method computes the scores in iterations as:
\begin{align}
\label{eq:pr}
\pr(t) = a \cdot \transit\cdot\pr(t-1) + \frac{1-a}{N} \mathbf{1}.
\end{align}
%We use a damping factor~$a$, while~$\mathbf{1}$ is the all-ones vector. 
We repeat this process until convergence, i.e., until $|\pr(t) - \pr(t-1)| \leq \epsilon$ for a small $\epsilon \geq 0$. We denote the PageRank score at the $i^{th}$ node as $\pr[i]$. For a sufficiently large number of steps $K$ for each state with $\Moves_k \!=\! \mathbf{1}\, \forall k \!\in\! [K]$, Equation~\eqref{eq:expreturn_detail} becomes $\ExpReward_{\policy}(k) = \RewardVector^{\top} \left( \Transitions^k \Initial \right)$. Likewise, for damping factor $a=1$, Equation~\eqref{eq:pr} becomes $\pr(t) =  \transit^t\pr(0)$, thus the two equations are rendered analogous with $\transit = \Transitions$ and~$\pr(0) = \Initial$. Then, considering that the iteration converges from step $\hat{k}$ onward, the expected reward from reward state~$\state_i$ per step $k \geq \hat{k}$, $\ExpReward_{\policy}(\{\state_i\}|k)$, is the PageRank score of the $i^\mathrm{th}$ node, that is~$\pr[i]$. To see this, let~$\RewardVector_i = \mathbf{1}_i$ be the reward vector when $\state_i \in \States$ is the \emph{only} reward state; then it holds that $\pr[i] = \mathbf{1}_i^{\top} \left(\transit^k \pr(0) \right) = \RewardVector_i^{\top} \left( \Transitions^k \Initial \right) = \ExpReward_{\policy}(\{\state_i\}|k)$.
%-------------------------------------------------------------

%Problem Formulation
%-------------------------------------------------------------
\section{Problem Formulation} 
In this section we model the uncertain environment where individuals navigate and introduce the \problem~(\problemshort) problem over a set of \modelnames (\modelshorts), extracted from real movement data, that express the behavior of individuals under different \settings.

\paragraph{Setting.}
Many applications generate data on the point-to-point movements of agents over a network, along with a distribution and their total number of steps. Using aggregate statistics on this information, we formulate, \emph{without loss of generality}, the movement of a population by a single agent moving probabilistically over the states of an \modelshort $\policy = (\States, \Initial, \Transitions, \Moves)$. Due to environment uncertainty, the agent may follow any of $|\Policies|$ different settings\footnote{We use the terms `setting' and `model' interchangeably.} $\Policies = \{\policy_1, \policy_2, \ldots, \policy_{|\Policies|}\}$.

\paragraph{\problem~Problem.} 
Several resource allocation problems can be formulated as optimization problems over an \modelshort~$\policy$, where reward states $\RewardStates$ correspond to the placement of resources. Given a budget $L$ and a cost function $c:\States\rightarrow\mathbb{N}^{+}$, the \problemnorobust~(\problemnorobustshort) problem seeks a set of reward states~$\RewardStates^* \subseteq \States$ that maximizes the cumulative reward~$\ExpReward(\RewardStates^*|\policy)$ obtained by an agent, that is: 
%$\RewardStates^* = \arg\max_{\RewardStates} \ExpReward(\RewardStates|\policy)$ s.t. $\sum_{\state \in \RewardStates} c[\state]$. 
\begin{align*}\label{eq:Ad-Placement}
\RewardStates^* = \arg\max_{\RewardStates} \ExpReward(\RewardStates|\policy)
\text{\quad s.t. } \sum_{\state \in \RewardStates} c[\state] \leq L.
\end{align*}
However, in reality the \emph{agent's} movements follow an unknown distribution sampled from a set of settings ${\Policies}~=~\{\policy_1, \policy_2, \ldots, \policy_{|\Policies|}\}$ represented as different \modelshorts.
Under this uncertainty, 
%given a set of \modelshorts~noted as~$\Policies$, 
the \problem~(\problemshort) problem seeks a set of reward states~\pki{\RewardStates}, within a budget, that maximizes the worst-case ratio of agent's cumulative reward over the optimal one, when the model $\policy\in\Policies$ is \emph{unknown}. 
%In particular, given a budget~$L$ and a cost function~$c: \States \rightarrow \mathbb{N}^{+}$, we seek a reward placement~\pki{$\RewardStates^* \subseteq \States$} such that:
Formally, we seek a reward placement~\pki{$\RewardStates^* \subseteq \States$} such that:
\begin{linenomath*}
\begin{equation}\label{eq:RRP}
\RewardStates^* = \arg\max_{\RewardStates} \min_{\policy \in \Policies} \frac{\ExpReward(\RewardStates|\policy)}{\ExpReward(\RewardStatesPolicy|\policy)} \text{\quad s.t. } \sum_{\state \in \RewardStates} c[\state] \leq L,
\end{equation}
\end{linenomath*}
where~$\RewardStatesPolicy = \arg\max\limits_{\RewardStates}{\ExpReward(\RewardStates|\policy)}$ is the optimal reward placement for a given model~$\policy \in \Policies$ within budget~$L$. 
This formulation is equivalent to minimizing the maximum \emph{regret ratio} of $\ExpReward(\RewardStates|\policy)$, i.e., $1 - \frac{\ExpReward(\RewardStates|\policy)}{\ExpReward(\RewardStatesPolicy|\policy)}$. 
The motivation arises from the fact that stakeholders are prone to compare what they achieve with what they could optimally achieve. The solution may also be seen as the optimal placement when the model~$\policy \in \Policies$ in which agents are moving is chosen by an \emph{omniscient} adversary, i.e. an adversary who chooses the setting~$\policy$ after observing the set of reward states~$\RewardStates$.
%-------------------------------------------------------------

%Hardness and Inapproximability Results
%-------------------------------------------------------------
\section{Hardness and Inapproximability Results}

In this section we examine the optimization problem of \problemshort~and we show that is \NPhard in general. First, in Theorem~\ref{th:rp_np_hard_Pi_1} we prove that even for a single model $(|\Policies|=1)$ the optimal solution cannot be found in polynomial time, due to a reduction from the \zok problem~\cite{karp1972reducibility}. 

\begin{theorem}\label{th:rp_np_hard_Pi_1}
The \problemshort~problem is \NPhard even for a single model, that is $|\Policies|=1$.
\end{theorem}
\begin{proof}
In the \zok problem~\cite{karp1972reducibility} we are given a set of items~$U$, each item~$u \in U$ having a cost~$c(u)$ and, wlog, an \emph{integer} value~$F(u)$ and seek a subset~$V \subseteq U$ 
that has total cost~$\sum_{v \in V}c(v)$ no more than a given budget~$L$ and maximum total value~$\sum_{v \in V}F(v)$. 
In order to reduce \zok to \problemshort, we set a distinct state $\state~\in~\States$ for each item~$u \in U$ with the same cost, i.e., $\States = U$, assign to each state a self-loop with transition probability~$1$, let each state be a reward state, and set a uniform initial distribution of agents over states equal to~$\sfrac{1}{\left|\States\right|}$ and steps probability equal to $\Moves[\state, k] = 1,\, \forall k \in [1,\ldots,F(u)]$. 
For a single setting, an optimal solution to the \problemshort~problem of Equation~\eqref{eq:RRP} is also optimal for the \NPhard \zok problem.\qedhere
\end{proof}

Theorem~\ref{th:inapprox} proves that \problemshort is inapproximable in polynomial time within constant factor, by a reduction from the \hitset problem, unless we exceed the budget constraint.

\begin{figure}[!b]
\centering
\includegraphics[width=0.37\textwidth]{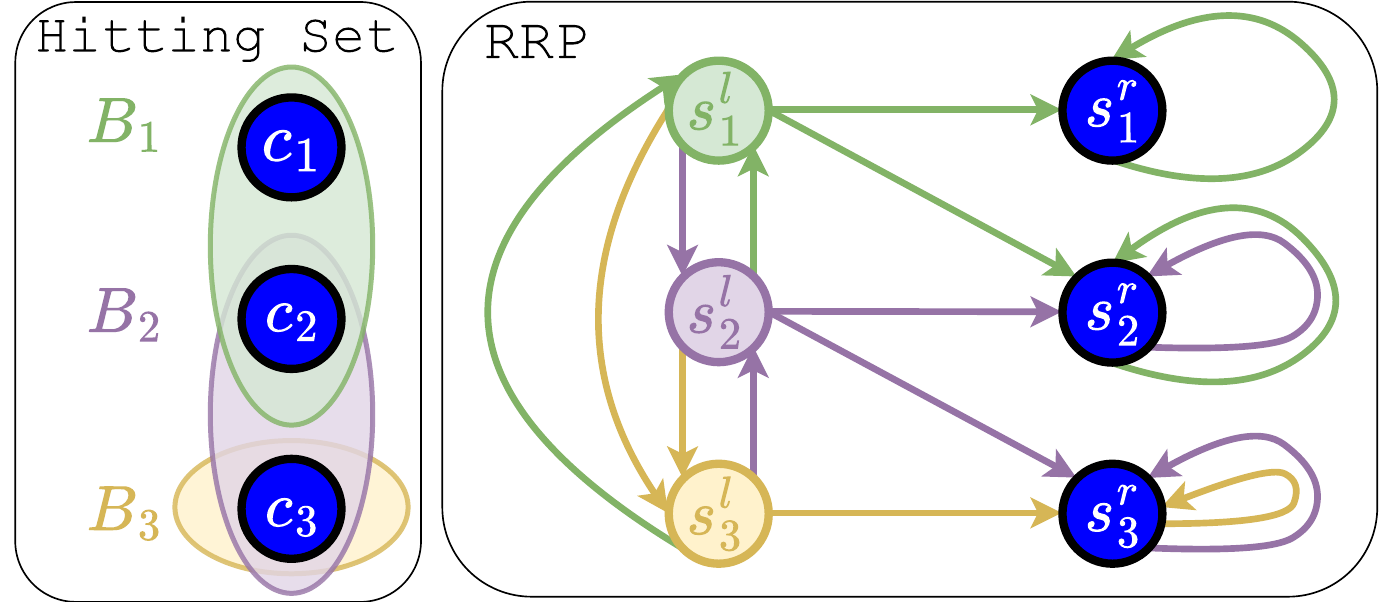}
\caption{\hitset (left) and \problemshort~reduction (right).
\label{fig:hitting_set}}
\end{figure}

\begin{theorem}\label{th:inapprox}
Given a budget~$L$ and set of models~$\Policies$, it is \NPhard to approximate the optimal solution to \problemshort within a factor of $\Omega(\sfrac{1}{n^{1 - \epsilon}})$, for any constant $\epsilon > 0$, unless the cost of the solution is at least~$\beta L$, with $\beta\geq \ln|\Policies|$.
\end{theorem}

\begin{proof}
We reduce the \pki{\hitset problem}~\cite{karp1972reducibility} to \problemshort and show that an approximation algorithm for \problemshort~implies one for \hitset. 
In the \hitset problem, given a collection of~$X$ items, $C = \{c_1, c_2, \ldots, c_X\}$ and a set of~$M$ subsets thereof, $B_i \subseteq C$, $i \in \{1, \ldots, M\}$, we seek a \emph{hitting set}~$C' \subseteq C$ such that~$B_i \cap C' \neq \emptyset \,\, \forall i \in \{1, \ldots, M\}$.

Given an instance of \hitset, we reduce it to \problemshort as follows. For each subset~$B_i$ we set a state $\state^l_i\in\States^l$ and for each item $c_i$ we set a state $\state^r_i\in\States^r$. Aslo, for each subset~$B_i$ we set an \modelshort $\policy_i$ $(|\Policies| = M)$ over the same set of states~$\States = \States^l\cup\States^r$ with $\States^l\cap\States^r=\emptyset$. We set the initial probabilities~$\Initial$ as uniform for all states in~$\States^l$, equal to~$\sfrac{1}{|\States^l|}$ for all models. Each model~$\policy_i \in \Policies$ features transition probabilities~$1$ from each state~$\state^l_j$ to state~$\state^l_i$, with $i \neq j$, and uniform transition probabilities from~$\state^l_i$ to each state~$\state^r_j$ if and only if~$c_j \in B_i$. 
States in~$\States^r$ are \emph{absorbing}, i.e., each state has a self-loop with probability $1$. Figure~\ref{fig:hitting_set} shows a small example of a \hitset instance and its \problemshort equivalent. We set the cost for absorbing states in $\States^r$ to $1$ and let each node in $\States^l$ have a cost exceeding~$L$. By this construction, if the reward placement~\RewardStates does not form a \emph{hitting set}, then there exists at least one subset~$B_i$, such that~$B_i \cap \RewardStates = \emptyset$, hence~$\min_{\policy}\frac{\ExpReward(\RewardStates|\policy)}{\ExpReward(\RewardStatesPolicy|\policy)} = 0$. In reverse, if~\RewardStates forms a hitting set, it holds that~$\min_{\policy} \frac{\ExpReward(\RewardStates|\policy)}{\ExpReward (\RewardStatesPolicy|\policy)} \geq \frac{1}{|\States^r|} > 0$. 
Thus, a hitting set exists if and only if~$\min_{\policy} \frac{\ExpReward(\RewardStates|\policy)}{\ExpReward (\RewardStatesPolicy|\policy)} > 0$. 
In effect, if we obtained an approximation algorithm for \problemshort by increasing the budget to~$\beta L$, for~$\beta > 1$, then we would also approximate, with a budget increased by a factor of~$\beta$, the \hitset problem, which is \NPhard for~$\beta < (1-\delta) \ln|\Policies|$ and~$\delta>0$ \cite{dinur2014analytical}.
\end{proof}
%-------------------------------------------------------------

%Connections to Knapsack Problems
%-------------------------------------------------------------
\section{Connections to Knapsack Problems}

In this section, we establish connections between RRP and \knapsack problems, which are useful in our solutions. 

\paragraph{Monotonicity and Additivity.} 
Lemma~\ref{lem:monot_addit} establishes that the cumulative reward function~$\ExpReward(\RewardStates|\policy)$ is monotone and additive with respect to~$\RewardStates$. These properties are vital in evaluating~$\ExpReward(\RewardStates|\policy)$ while exploiting pre-computations.

\begin{lemma}\label{lem:monot_addit}
The cumulative reward~$\ExpReward(\RewardStates|\policy)$ in Equation~\eqref{eq:expreturn} is a monotone and additive function of reward states~\RewardStates.
\end{lemma}
\begin{proof}
By Equation~\eqref{eq:expreturn} we obtain the monotonicity property of the cumulative reward function $\ExpReward(\cdot|\policy)$. 
Given a model $\policy\in\Policies$ and two sets of reward states $\RewardStatesA\subseteq \RewardStatesB\subseteq \States$ every term of $\ExpReward(\RewardStatesA|\policy)$ is no less than its corresponding term of $\ExpReward(\RewardStatesB|\policy)$ due to Equation~\eqref{eq:expreturn_detail}. 
For the \emph{additivity} property it suffices to show that any two sets of reward states $\RewardStatesA,\RewardStatesB\subseteq\States$ satisfy:
\begin{align*}
\ExpReward(\RewardStatesA|\policy)~+~\ExpReward(\RewardStatesB|\policy)~=~\ExpReward(\RewardStatesA~\cup~\RewardStatesB|\policy)~+~\ExpReward(\RewardStatesA\cap\ \RewardStatesB|\policy).
\end{align*}
%At time~$t=0$, 
Assume w.l.o.g. that the equality holds at time~$t$, i.e.
$\currentreward^t_{\RewardStatesA} + \currentreward^t_{\RewardStatesB} = \currentreward^t_{\RewardStatesAcapB} + \currentreward^t_{\RewardStatesAcupB}$, $\currentreward^t_{\RewardStatesX}$ being the cumulative reward at time~$t$ for 
%the set of 
reward states $\RewardStatesX$.
%Assume wlog that the equality holds for time $t$.
It suffices to prove that the additivity property holds for~$t+1$. At timestamp~$t+1$, the agent at state~$\state \in \States$ moves to~$\state' \in \States$. We distinguish cases as follows:
\begin{enumerate}[noitemsep,topsep=0pt,parsep=0pt,partopsep=0pt]
    \item 
    If $\state'\notin \RewardStatesA\cup\RewardStatesB$ then $\state'\notin \RewardStatesA\cap\RewardStatesB$, $\state'\notin \RewardStatesA$ and $\state'\notin \RewardStatesB$, thus additivity holds.
    
    \item 
    If $\state'\in \RewardStatesA\cup\RewardStatesB$ and $\state'\notin \RewardStatesA\cap\RewardStatesB$ then either $\state'\in \RewardStatesA$ or $\state'\in \RewardStatesB$.
    Assume wlog that $\state'\in \RewardStatesA$, then it holds that:
    $\currentreward^{t+1}_{\RewardStatesA} = \currentreward^{t}_{\RewardStatesA} + \Transitions[\state, \state']$, 
    $\currentreward^{t+1}_{\RewardStatesAcupB} = \currentreward^{t}_{\RewardStatesAcupB} + \Transitions[\state,\state']$, 
    $\currentreward^{t+1}_{\RewardStatesB} = \currentreward^{t}_{\RewardStatesB}$ and $\currentreward^{t+1}_{\RewardStatesAcapB} = \currentreward^{t}_{\RewardStatesAcapB}$.
    
    \item If $\state'\in \RewardStatesA\cap\RewardStatesB$ then $\state'\in \RewardStatesA$ and $\state'\in \RewardStatesB$.
    Then, it holds that:
    $\currentreward^{t+1}_{\RewardStatesA} = \currentreward^{t}_{\RewardStatesA} + \Transitions[\state, \state']$, $\currentreward^{t+1}_{\RewardStatesB} = \currentreward^{t}_{\RewardStatesB} + \Transitions[\state, \state']$, $\currentreward^{t+1}_{\RewardStatesAcupB}~=~\currentreward^{t}_{\RewardStatesAcupB}~+~\Transitions[\state, \state']$, and~$\currentreward^{t+1}_{\RewardStatesAcapB} = \currentreward^{t}_{\RewardStatesAcapB} + \Transitions[\state, \state']$.
\end{enumerate}
In all cases the cumulative reward function is additive.
\end{proof}

Next, Lemma~\ref{lem:rp_optimal_linear} states that \problemshort under a single model~$\policy$ $\left(|\Policies| = 1\right)$, i.e., the maximization of~$\ExpReward(\RewardStates|\policy)$ within a budget~$L$, is solved in pseudo-polynomial time thanks to the additivity property in Lemma~\ref{lem:monot_addit} and a reduction from the \zok problem~\cite{karp1972reducibility}.
Lemma~\ref{lem:rp_optimal_linear} also implies that we can find the optimal reward placement with the maximum expected reward by using a single expected setting $\policy$. 

\begin{restatable}{lemma}{lemoptimalsingle}
\label{lem:rp_optimal_linear}
For a single model~$\policy$ $\left(|\Policies| = 1\right)$ and a budget~$L$, there is an optimal solution for \problemshort~that runs in pseudo-polynomial time~$\mathcal{O}(L n)$.
\end{restatable}
%\iffalse
\begin{proof}
For each state~$s_i \in \States$ we set an item~$u_i \in U$ with cost~$c(u_i) = c[\state_i]$ and value~$F(u_i) = \ExpReward(\{s_i\}|\policy)$. Since the reward function is additive (Lemma~\ref{lem:monot_addit}), it holds that $\ExpReward( \RewardStates|\policy) = \sum_{s_i \in \RewardStates}{\ExpReward(\{s_i\}|\policy)} = \sum_{u_i \in U}F(u_i)$. Thus, we can optimally solve single setting \problemshort~in pseudo-polynomial time by using the dynamic programming solution for \zok~\cite{martello1987algorithms}.
\end{proof}
%\fi

In the \mnk problem (MNK), given a set of items~$U$, each item~$u \in U$ having a cost~$c(u)$, and a collection of scenarios~$X$, each scenario~$x \in X$ having a value~$F_{x}(u)$, we aim to determine a subset~$V~\subseteq~U$,  
%that has total cost no more than a given budget~$L$ 
with total cost no more than $L$, 
and maximizes the minimum total value across scenarios, i.e., $\arg_V \max \min_x \sum_{u \in V} F_x(u)$. The following lemma reduces the \problemshort~problem to \mnk~\cite{yu1996max} in pseudo-polynomial time.

\begin{restatable}{lemma}{lemrpminmax} \label{lem:rp_and_minmax}
\problemshort is reducible to \mnk in~$\mathcal{O}(|\Policies|Ln)$ time.
\end{restatable}
%-------------------------------------------------------------

%Approximation Algorithm
%-------------------------------------------------------------
%\section{Bicriteria Optimal Robust Reward Placement}
\section{Approximation Algorithm}

Here, we introduce $\Psi$-Saturate,\footnote{$\Psi$ for `pseudo-', from Greek `\textgreek{ψευδής}'.} a pseudo-polynomial time binary-search algorithm based on the Greedy-Saturate method \cite{rimkempe}. For any $\epsilon>0$, $\Psi$-Saturate returns an $\epsilon$-additive approximation of the optimal solution by exceeding the budget constraint by a factor $\mathcal{O}(\ln |\Pi|/\epsilon)$.

\begin{algorithm}[!h]
\caption{$\Psi$-Saturate Algorithm}
\label{alg:pseudo_sat}
\textbf{Input}: \modelshorts~$\Policies$, steps~$K$, budget~$L$, precision~$\epsilon$, param.~$\beta$.\\
\textbf{Output}: 
%Optimal~
Reward~Placement~\pki{$\RewardStates^*$~}of~cost~at~most~$\beta L$.
\vspace*{-4mm}
\begin{algorithmic}[1] %[1] enables line numbers
\FOR{$\policy \in \Policies$} \label{alg:for_loop_models}
\STATE
\pki{$\RewardStatesPolicy \gets \text{Knapsack}(\policy, L)$} 
\label{alg:knapsack_single_model}
\ENDFOR
\STATE $\eta_{min} \gets 0$, $\eta_{max} \gets 1$, \Petros{$\RewardStates^*\gets \emptyset$}
\label{alg:etas_init}
\WHILE {($\eta_{min} - \eta_{max})\geq\epsilon$}
\label{alg:while_binary_search}
\STATE $\eta \gets (\eta_{max} + \eta_{min}) / 2$
\label{alg:eta_update}
%\STATE $\Rewards \gets \emptyset $
\STATE\Petros{$\RewardStates \gets \emptyset $}
\label{alg:state_init}
\WHILE {$\sum\limits_{\policy\in\Policies}\min\left(\eta,\frac{\ExpReward(\Petros{\RewardStates}|\policy
)}{\ExpReward(\RewardStatesPolicy|\policy)}\right) < (\eta\cdot|\Policies| - \eta\cdot\epsilon/3)$}
\label{alg:while_argmax_search}
\STATE $\state~\leftarrow~\argmax\limits_{\state\in \States\backslash
%\Rewards
\Petros{\RewardStates}
}\sum\limits_{\policy\in\Policies}\frac{1}{c(s)}\Big(\min\left(\eta,\frac{\ExpReward(\Petros{\RewardStates}\cup\{\state\}|\policy)}{\ExpReward(\Petros{\RewardStatesPolicy}|\policy)}\right)-$\\
\hspace*{-19.5mm}$\quad\quad\qquad\qquad\qquad\qquad\qquad\qquad\qquad\min\left(\eta,\frac{\ExpReward(\Petros{\RewardStates}|\policy)}{\ExpReward(\RewardStatesPolicy|\policy)}\right)\Big)$
%\STATE $\Rewards \gets \Rewards \cup \{\state\}$
\label{alg:argmax_state}
\Petros{
\STATE $\RewardStates \gets \RewardStates \cup \{\state\}$
\label{alg:solution_add_state}
}
\ENDWHILE
\IF{$\sum_{\state\in
%\Rewards
\Petros{\RewardStates}
}c[\state]>\beta  L$}
\label{alg:check_exceeding_budget_constraint}
\STATE  $\eta_{max}\gets \eta$
\label{alg:eta_max_update}
\ELSE
\label{alg:else_budget_constraint}
\STATE $\eta_{min} \gets \eta\cdot (1-\epsilon /3)$
\label{alg:eta_min_update}
%\STATE $\Rewards^*\gets \Rewards$
\STATE \Petros{$\RewardStates^*\gets \RewardStates$}
\label{alg:solution_update}
\ENDIF
\ENDWHILE
\label{alg:endwhile}
\STATE \textbf{return} $\RewardStates^*$
\label{alg:solution_return}
\end{algorithmic}
\end{algorithm}

\paragraph{The $\Psi$-Saturate Algorithm.}
Algorithm~\ref{alg:pseudo_sat} presents the pseudocode of $\Psi$-Saturate.
As a first step, in Lines~\ref{alg:for_loop_models}--\ref{alg:knapsack_single_model}, the algorithm finds the optimal reward placement $\RewardStatesPolicy$ for each model $\policy\in\Policies$; this is needed for evaluating the denominator of the RRP objective value in Equation~\eqref{eq:RRP}.
By Lemma~\ref{lem:rp_optimal_linear}, $\RewardStatesPolicy$ is computed in pseudo-polynomial time using the dynamic programming algorithm for the \knapsack problem.
Then, in Lines~\ref{alg:while_binary_search}--\ref{alg:endwhile} the algorithm executes a binary search in the range of the \minmax objective ratio (Line~\ref{alg:etas_init}).
In each iteration, the algorithm makes a guess $\eta$ of the optimal \minmax objective value (Line~\ref{alg:eta_update}), and then seek a set of reward states $\RewardStates$ (Line~\ref{alg:state_init}), of minimum cost, with score at least $\eta$ (Line~\ref{alg:while_argmax_search}), within distance $\epsilon>0$.
Finding $\RewardStates$ of the minimum cost, implies an optimal solution for the NP-hard RRP problem.
Thus, in Lines~\ref{alg:argmax_state}--\ref{alg:solution_add_state}, $\Psi$-Saturate approximates this solution by using the Greedy algorithm in \cite{wolsey1982analysis} 
for function $\min\left(\eta,\frac{\ExpReward(\Petros{\RewardStates}|\policy)}{\ExpReward(\RewardStatesPolicy|\policy)}\right)$ which, for fixed $\policy$ and $\eta$, is monotone and submodular.\footnote{The minimum of a constant function $\left(\eta\right)$ and a monotone additive function $\left(\frac{\ExpReward(\Petros{\RewardStates}|\policy
)}{\ExpReward(\RewardStatesPolicy|\policy)},  \text{Lemma~\ref{lem:monot_addit}}\right)$ is monotone and submodular. The term ${\ExpReward(\RewardStatesPolicy|\policy)}$ is constant as it has been computed in Line~\ref{alg:knapsack_single_model}.}
%In Lines~\ref{alg:argmax_state}--\ref{alg:solution_add_state}, $\Psi$-Saturate evaluates function $\min\left(\eta,\frac{\ExpReward(\Petros{\RewardStates}|\policy)}{\ExpReward(\RewardStatesPolicy|\policy)}\right)$ which, for fixed $\policy$ and $\eta$, is monotone and submodular\footnote{The minimum of a constant function $\left(\eta\right)$ and a monotone additive function $\left(\frac{\ExpReward(\Petros{\RewardStates}|\policy)}{\ExpReward(\RewardStatesPolicy|\policy)},  \text{Lemma~\ref{lem:monot_addit}}\right)$ is monotone and submodular. The term ${\ExpReward(\RewardStatesPolicy|\policy)}$ is constant as it has been computed in Line~\ref{alg:knapsack_single_model}.}, by using the Greedy approximation algorithm of \cite{wolsey1982analysis}. 
If the formed solution exceeds the budget constraint, the algorithm decreases the upper bound of the search scope (Lines~\ref{alg:check_exceeding_budget_constraint}--\ref{alg:eta_max_update}), otherwise it increases the lower bound and updates the optimal solution $\RewardStates^*$ (Lines~\ref{alg:else_budget_constraint}--\ref{alg:solution_update}). 
Finally, it returns the optimal solution found (Line~\ref{alg:solution_return}).

\iffalse
Following an analogous proof to Theorem~3 in work of~\cite{rimkempe}, we derive Lemma~\ref{lem:bicriteria_apx} which states that $\Psi$-Saturate approximates the optimal value within distance $\epsilon$ 
%when it exceeds the budget by a factor $\mathcal{O}(\ln\sfrac{|\Policies|}{\epsilon})$, 
by exceeding budget to $\beta L$, with $\beta = 1 + \ln\frac{3|\Policies|}{\epsilon}$,
i.e., offers an ($\text{OPT}-\epsilon$, $\beta L$) bicriteria approximation solution.
\fi

In Lemma~\ref{lem:bicriteria_apx} we prove that by setting $\beta= 1 + \ln\frac{3|\Policies|}{\epsilon}$, $\Psi$-Saturate approximates the optimal value within distance $\epsilon$.

\begin{restatable}{lemma}{lembicriteria}\label{lem:bicriteria_apx}
For any constant $\epsilon > 0$, let $\beta = 1 + \ln\frac{3|\Policies|}{\epsilon}$. 
$\Psi$-Saturate finds a reward placement~\RewardStates of cost at most $\beta L$ with~$\min_{\policy} \frac{\ExpReward(\RewardStates |\policy)}{\ExpReward(\RewardStatesPolicy|\policy)} \geq \min_{\policy} \frac{\ExpReward(\RewardStates^*|\policy )}{\ExpReward(\RewardStatesPolicy |\policy )} - \epsilon={OPT}-\epsilon$, and $\RewardStates^*~=~\arg_{\RewardStates} \max \min_\policy \frac{\ExpReward( \RewardStates|\policy )}{\ExpReward(\RewardStatesPolicy|\policy  )} \text{  s.t. } \sum_{\state \in \RewardStates} c[\state] \leq L$.
\end{restatable}
\begin{proof}
%The proof is analogous to Theorem~3 in~\cite{rimkempe}.
We seek to solve a~\maxmin regret optimization problem of an \emph{additive function under a knapsack constraint}.
While finding the optimal score in the denominator of the~\maxmin ratio is \NPhard due to Theorem~\ref{th:rp_np_hard_Pi_1}, in Line~\ref{alg:knapsack_single_model} we evaluate it by a pseudo-polynomial time Knapsack algorithm, as Lemma~\ref{lem:rp_optimal_linear} allows.
In Lines~\ref{alg:while_binary_search}--\ref{alg:endwhile}, we perform a binary search to find a reward placement of cost at most~$\beta L$.
By the analysis in~\cite{rimkempe}, the~$\Psi$-Saturate algorithm provides an $\left(\beta,\text{OPT}-\epsilon\right)$ bicriteria approximation for \problemshort, where \text{OPT} is the optimal objective ratio score.
\end{proof}

Unlike the pseudo-polynomial-time dynamic programming algorithm (Knapsack, Line~\ref{alg:knapsack_single_model}) we employ, the Greedy-Saturate algorithm~\cite{rimkempe} uses the Greedy\footnote{\label{notegreedy} The algorithm iteratively selects the element, within the budget, that offers the maximal marginal gain divided by its cost.} algorithm to \emph{approximate} the optimal reward placement~$\RewardStatesPolicy$ (Lines~\ref{alg:for_loop_models}--\ref{alg:knapsack_single_model}), which provides an~$\sfrac{1}{2}$-approximation of the optimal solution for a monotone additive function over a knapsack constraint~\cite{johnson1979computers}. As our reward function is monotone and additive (Lemma~\ref{lem:monot_addit}), Greedy-Saturate offers an~($\frac{1}{2}\text{OPT}-\epsilon$, $\beta L$) bicriteria approximation.
\iffalse
\begin{corollary}\label{cor:greedysat}
For any constant $\epsilon > 0$, let $\beta = 1 + \ln\frac{3|\Policies|}{\epsilon}$. 
Greedy-Saturate finds a reward placement~\RewardStates of cost at most $\beta L$ with $\min_{\policy} \frac{\ExpReward(\RewardStates|\policy )}{\ExpReward(\RewardStatesPolicy|\policy )} \geq \frac{1}{2}\min_{\policy} \frac{\ExpReward(\RewardStates^* |\policy  )}{\ExpReward( \RewardStatesPolicy| \policy  )} - \epsilon=\frac{1}{2}\text{OPT}-\epsilon$, and $\RewardStates^* = \arg_{\RewardStates} \max \min_\policy \frac{\ExpReward( \RewardStates|\policy  )}{\ExpReward(\RewardStatesPolicy|\policy )} \text{ s.t. } \sum_{\state \in \RewardStates} c[\state] \leq L$.
\end{corollary}
\fi

Notably, for~$\beta = 1$, $\Psi$-Saturate returns an non-constant approximation of the optimal solution within the budget constraint $L$. In particular, the next corollary holds.

\begin{corollary}\label{cor:apx}
For any constant $\epsilon>0$, let $\gamma= 1+\ln\frac{3|\Policies|}{\epsilon}$. 
For $\beta=1$, $\Psi$-Saturate satisfies the budget constraint and returns an $\frac{1}{\gamma}(OPT'-\epsilon)$ approximation \Petros{factor} of the optimal solution, with $OPT'~=~\max_{\RewardStates}\min_\policy{\frac{\ExpReward( \RewardStates|\policy )}{\ExpReward(\RewardStatesPolicy |\policy)}}$ s.t. $\sum_{\state \in \RewardStates} c[\state] \leq \frac{L}{\gamma}$.
\end{corollary}

We stress that the approximation in Corollary~\ref{cor:apx} is non-constant and can be arbitrarily small, as implied by the inapproximability result of Theorem~\ref{th:inapprox}. %However, the corollary indicates that if the optimal value for a smaller budget constraint is non-zero, then $\Psi$-Saturate for $\beta=1$ provides an approximation of the optimal solution within the initial budget constraint $L$.
%-------------------------------------------------------------

%Heuristic Solutions
%-------------------------------------------------------------
\section{Heuristic Solutions}

Inspired from previous works on node selection in networks~\cite{rimkempe,zhang2020geodemographic} and the connection of \problemshort with Knapsack problems, we propose four heuristic methods.
For a single model $\left( |\Policies| = 1 \right)$ and under uniform costs~$\left(c[\state] = c \,\, \forall \state \in \States\right)$, these four heuristics find an optimal solution. However, contrary to $\Psi$-Saturate algorithm (Lemma~\ref{lem:bicriteria_apx}), they may perform arbitrarily badly in the general multi-model case, even by exceeding the budget constraint. To accelerate the selection process, we use the \emph{Lazy Greedy} technique that updates values selectively~\cite{minoux1978} in all heuristics, except the one using dynamic programming.

\paragraph{All Greedy.} The All Greedy method optimally solves the \problemshort~problem for each model~$\policy \in \Policies$ separately using the Knapsack dynamic programming algorithm (Lemma~\ref{lem:rp_optimal_linear}) and then picks, among the collected solutions, the one yielding the best value of the objective in Equation~\eqref{eq:RRP}. All Greedy is optimal for a single model with an arbitrary cost function.

\paragraph{Myopic.} A greedy algorithm that iteratively chooses the reward state $\state^*\in\States$, within the budget, that offers the maximal marginal gain ratio to the \problemshort~objective divided by the cost, that is~$\state^* = \argmax \limits_{\state \in \States\backslash \RewardStates} \min \limits_{ \policy \in \Policies} \left( \frac{1}{c[s]} \frac{ \ExpReward(\RewardStates\cup\{s\} | \policy) - \ExpReward(\RewardStates|\policy)} {\ExpReward( \RewardStatesPolicy | \policy)} \right)$.

\paragraph{Best-Worst Search (BWS).}
This algorithm uses as a score the minimum, over settings, cumulative reward for a set~$\RewardStates$, that is~$H(\RewardStates) = \min_{\policy}\ExpReward( \RewardStates | \policy)$ and iteratively chooses the reward state~$\state^* \in \States$, within the budget, that offers the maximal marginal gain to that score divided by the cost, that is $\state^* = \argmax \limits_{\state\in \States\backslash \RewardStates} \left( \frac{H( \RewardStates \cup\{s\}) - H(\RewardStates)}{c[s]}\right)$.

\paragraph{Dynamic Programming (DP-RRP).}
In Lemma~\ref{lem:rp_and_minmax} we reduced \problemshort to \mnk (MNK) in pseudo-polynomial time. While MNK admits an optimal solution using a pseudo-polynomial time dynamic programming algorithm, its running time grows exponentially with the number of settings $|\Policies|$ \cite{yu1996max}.
To overcome this time overhead, we propose a more efficient albeit \emph{non-optimal} dynamic-programming algorithm for the RRP problem, noted as DP-RRP. 
For reward placement $\RewardStates$, we denote the cumulative reward for each setting as the following $|\Policies|$-tuple: $g(\RewardStates)~=~\left(\ExpReward(\RewardStates|\policy_1), \ExpReward(\RewardStates|\policy_2), \ldots, \ExpReward(\RewardStates|\policy_{|\Policies|}) \right)$.
We use an $(n+1) \times (L+1)$ matrix $M$ whose entries are $|\Policies|$-tuples of the form $g(\cdot)$.
Let $\min g(\RewardStates)=\min_{\policy_i} \ExpReward( \RewardStates|\policy_i)$ be the minimum reward, across~$|\Policies|$ settings.
We define the maximum of two entries $g(\RewardStatesone)$ and $g(\RewardStatestwo)$, as $\arg\max_{\RewardStates \in \{ \RewardStatesone, \RewardStatestwo\}} \min g(\RewardStates)$, i.e. the one holding the largest minimum reward.
We initialize $M[\cdot, 0]~=~M[0, \cdot] = (0,0, \ldots,0)$ 
and recursively compute $M{[i, j]}$~as follows:
\begin{linenomath*}
\begin{equation}\label{eq:Mrec}
M[i,j] = \max\{M[i\!-\!1,j], M[i\!-\!1, j\!-\!c[i]] \!+\! g(\{i\})\},
\end{equation}
\end{linenomath*}
where~$M[i,j]$ stands for a solution using the first~$i$ states, by some arbitrary order, and~$j$ units of budget. 
In the recursion of 
Equation~\eqref{eq:Mrec}, the first option stands for \emph{not} choosing state~$\state_i$ as a reward state, while the latter option stands for doing so while paying cost~$c[i]$ and gaining the additive reward~$g(\{i\})$. We compute~$M[n, L]$ as above in space and time complexity~$\Theta(|\Policies|Ln)$ and backtrack over~$M$ to retrieve the selected reward states in the final solution. 
Note that, for a single model, i.e. $|\Policies|=1$ and arbitrary cost function, Equation~\eqref{eq:Mrec} returns an optimal solution.

\paragraph{Worst-case performance.} While all heuristics\footnote{All algorithms work, without modification, with rewards of arbitrary non-negative values and when a partial solution is given.} approach the optimal solution under a single setting, they may perform arbitrarily badly with multiple settings. In Lemma~\ref{lem:lemheur} we prove that this holds even when exceeding the budget constraint, contrariwise to the $\Psi$-Saturate algorithm (Lemma~\ref{lem:bicriteria_apx}).

\begin{restatable}{lemma}{lemheur}\label{lem:lemheur}
The heuristics for \problemshort may perform arbitrarily badly even when they exceed the budget constraint from~$L$ to~$\beta L$, with~$\beta = 1 + \ln\frac{3|\Policies|}{\epsilon}$ and~$\epsilon>0$.
\end{restatable}

%\paragraph{Extensions.} All algorithms work, without modification, with rewards of arbitrary non-negative values, i.e. $\RewardVector[\cdot]\in \mathbb{R}^{+}$, and when a partial solution is already given.
%-------------------------------------------------------------

%Experimental Analysis
%-------------------------------------------------------------
\section{Experimental Analysis}

In this section we evaluate the running time and performance of algorithms on synthetic and real-world data. We use different problem parameters as Table~\ref{tbl:dataset} shows, marking the default value of each parameter in bold. To satisfy the budget constraint for the $\Psi$-Saturate algorithm, we fix~$\beta =1$ as in Corollary~\ref{cor:apx} and set precision to~$\epsilon = ({|\Policies | \cdot10^{3}})^{-1}$. We set the budget~$L$ as a percentage of the total cost~$\sum_{\state\in\States}c[s]$. To benefit from the additivity property of Lemma~\ref{lem:monot_addit}, we precompute the cumulative reward $\ExpReward(\{\state\}|\policy)$ for each state $\state\in\States$ and model $\policy\in\Policies$. We implemented\footnote{ \url{https://anonymous.4open.science/r/RRP-F6CA}} all methods in C++~17 and ran experiments on a~376GB server with~96 CPUs @2.6GHz.

\begin{table}[!h]
%\vspace{-2mm}
\centering
\begin{tabular}{ |c||c| }
% \hline
%\multicolumn{2}{|c|}{Dataset Statistics} \\
 \hline
 Parameter & Values\\
 \hline
$n$   &   $2500$, $5000$, $7500$, \pmb{$10000$}, $12500$  \\ 
 \hline
  $\langle d\rangle$ &  3,\textbf{6},9,12  \\
\hline
 %\multicolumn{2}{|c|}{\scfr} \\
 %\hline
 $p_{\beta}$ &  0.6, 0.7, \textbf{0.8}, 0.9  \\
 \hline
 $|\Policies|$ &  2,5,\textbf{10},15,20  \\
 \hline
 $K$  & 2,4,\textbf{6},8,10\\
 \hline
 $L$    & 10\%, \textbf{25\%}, 50\%, 75\% \\
 %\hline
 %Alpha &  \textbf{1}
 %,1.25,1.5,1.75,2 \Petros{Remove}  
 %\\
 \hline
 %\multicolumn{2}{|c|}{\er} \\
 %\hline
\end{tabular}
%\vspace{-2.5mm}
\caption{Parameter settings.}\label{tbl:dataset}
\vspace{-2mm}
\end{table}

\subsection{Synthetic Data}

We use two different types of synthetic datasets to represent stochastic networks (i.e., \modelshorts). In each type, we generate a directed graph and then sample edge weights from a normal distribution to create different settings. In more detail:

\paragraph{\er:} We generate~$20$ directed graphs for each of the~$5$ sizes shown in Table~\ref{tbl:dataset}. In all cases, we set the edge creation probability to achieve the desired average in-degree~$\langle d\rangle$. %(default~$6$).

\paragraph{\scfr:} We generate~$20$ directed scale-free graphs for each of the~$5$ sizes shown in Table~\ref{tbl:dataset}. Following~\cite{bollobas2003directed}, we use three parameters to construct the network:~$p_{\alpha}$ ($p_{\gamma}$), the probability to add a new node connected to an existing one chosen randomly by its in-degree (out-degree), and~$p_\beta$, the probability to add an edge~$(u,v)$, with~$u$ and~$v$ selected by their out-degree and in-degree respectively. In all datasets, we tune~$p_{\beta}$ and set~$p_{\gamma} = \frac{1 - p_{\beta}}{3}$ and~$p_{\alpha} = 2p_{\gamma}$, so that $p_{\alpha} + p_{\beta} + p_{\gamma} = 1$.

Given a graph structure, we generate~$|\Policies|=20$ distinct settings, corresponding to different models. For each setting~$\policy_i$, we sample the weight of edge~$(u,v)$ from a normal distribution with mean~$\mu = \sfrac{1}{d_u}$ and standard deviation~$\sigma_i = \sfrac{i}{10d_u}$. When we sample a negative value, we set the edge weight to zero. In each resulting directed graph, we set transition probabilities~$\Transitions$ as normalized edge weights. Moreover, we set the initial probabilities~$\Initial$ proportionally to the sum of nodes' outgoing weights and the cost of a node as the rounded-down average number of its in-neighbors among settings.

\pgfplotsset{width=3.7cm,compat=1.18}
%\vspace{-3mm}
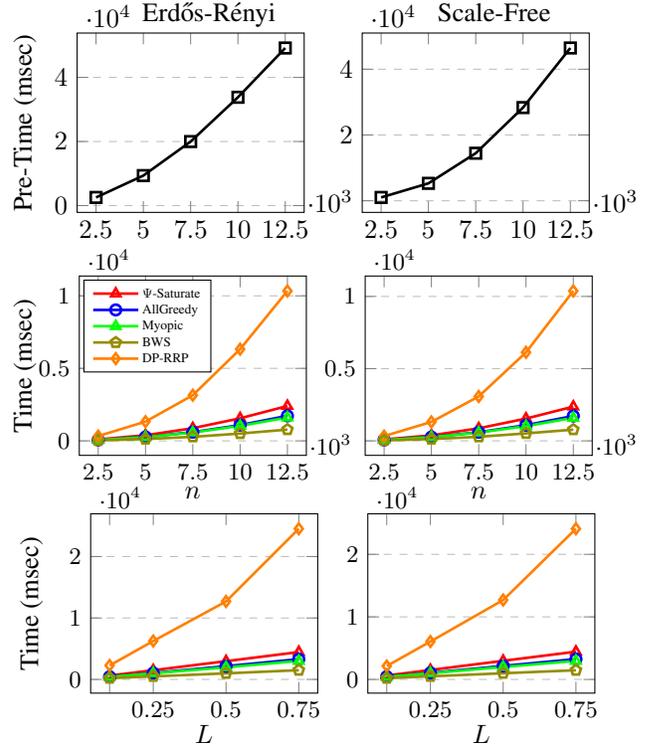
\begin{figure}[!h]
%\hspace*{-30mm}
\begin{subfigure}{.45\linewidth}
\pgfplotsset{width=4.6cm,compat=1.18}
\begin{tikzpicture}
\begin{axis}[
    title={\er},
    title style={yshift=-7 pt,},
    title style={xshift=7 pt,},
    %title style = {font=\small},
    %xlabel = \(N\),
    ylabel = {Pre-Time (msec)},
    %xmode = log,
    %ymode = log,
    %xmin=392, xmax=83772,
    %xtick={392,719,1185,1616,2224,3199,4845,9015,83772},
    legend pos=north west,
    legend style={nodes={scale=0.7, transform shape}},
    ymajorgrids=true,
    grid style=dashed,
    %ticklabel style = {font=\tiny},
    %xlabel style={font=\tiny},
    %ylabel style={font=\tiny},
    ylabel shift = -0 pt,
    %xlabel shift = -5 pt,
    %yticklabel shift = -1 pt,
    %xticklabel shift = -1 pt,
    legend cell align={left},
    xtick={2500, 5000, 7500, 10000, 12500},
    ytick={0,20000,40000},
    yticklabels={$0$,$2$,$4$},
    every x tick scale label/.style={
    at={(1,0)},xshift=1pt,anchor=south west,inner sep=0pt
    },
    scaled x ticks=base 10:-3,
    %log basis y={2},
    scaled y ticks=base 10:-4
]
\addplot[
    color=black,
    mark=square,
    mark options=solid,
    mark size=2pt,
    line width=1pt
    ] table[x={N},  y = {PreTime}]{Data/Random/N/allGreedy.txt}; ]
%\legend{Saturate, AllGreedy, Random}
\end{axis}
\end{tikzpicture}
\end{subfigure}
\hspace*{3mm}
\begin{subfigure}{.45\linewidth}
\pgfplotsset{width=4.6cm,compat=1.18}
\begin{tikzpicture}
\begin{axis}[
    title={\scfr},
    title style={yshift=-5 pt,},
    title style={xshift=7 pt,},
    %title style = {font=\small},
    %xlabel = \(N\),
    %ylabel = {Time (msec)},
    %xmode = log,
    %ymode = log,
    %xmin=392, xmax=83772,
    %xtick={392,719,1185,1616,2224,3199,4845,9015,83772},
    legend pos=north west,
    legend style={nodes={scale=0.7, transform shape}},
    ymajorgrids=true,
    grid style=dashed,
    %ticklabel style = {font=\tiny},
    %xlabel style={font=\tiny},
    %ylabel style={font=\tiny},
    %ylabel shift = -6 pt,
    %xlabel shift = -5 pt,
    %yticklabel shift = -1 pt,
    %xticklabel shift = -1 pt,
    legend cell align={left},
    xtick={2500, 5000, 7500, 10000, 12500},
    ytick={0,20000,40000},
    yticklabels={,$2$,$4$},
    every x tick scale label/.style={
    at={(1,0)},xshift=1pt,anchor=south west,inner sep=0pt
    },
    scaled x ticks=base 10:-3,
    %log basis y={2},
    scaled y ticks=base 10:-4
]
\addplot[
    color=black,
    mark=square,
    mark options=solid,
    mark size=2pt,
    line width=1pt
    ] table[x={N},  y = {PreTime}]{Data/ScaleFree/N/allGreedy.txt}; ]
%\legend{Saturate, AllGreedy, Random}
\end{axis}
\end{tikzpicture}
\end{subfigure}
\\
\begin{subfigure}{.45\linewidth}
\pgfplotsset{width=4.6cm,compat=1.18,every axis legend/.append style={
at={(0.3,0.98)},
anchor=north}}
\begin{tikzpicture}
\begin{axis}[
    title style={yshift=-5 pt,},
    xlabel = \(n\),
    ylabel = {Time (msec)},
    %xmode = log,
    %xmin=392, xmax=83772,
    %xtick={392,719,1185,1616,2224,3199,4845,9015,83772},
    %legend pos=north west,
    legend style={nodes={scale=0.5, transform shape}},
    ymajorgrids=true,
    grid style=dashed,
    ,scaled y ticks={base 10:-2},
    %ticklabel style = {font=\tiny},
    %xlabel style={font=\tiny},
    %ylabel style={font=\tiny},
    ylabel shift = -6 pt,
    %xlabel shift = -5 pt,
    %yticklabel shift = -1 pt,
    %xticklabel shift = -1 pt,
    legend cell align={left},
    ticklabel style = {font=\small},
    %xlabel style={font=\tiny},
    %ylabel style={font=\tiny},
    %ylabel shift = -10 pt,
    xlabel shift = -4 pt,
    xtick={2500, 5000, 7500, 10000, 12500},
    every x tick scale label/.style={
    at={(1,0)},xshift=1pt,anchor=south west,inner sep=0pt
    },
    scaled x ticks=base 10:-3,
    scaled y ticks=base 10:-4
]
\addplot[
    color=red,
    mark=triangle,
    mark options=solid,
    mark size=2pt,
    line width=1pt
    ] table[x={N},  y = {Time}]{Data/Random/N/saturate.txt}; ]
%\addplot[
%    color=brown,
%    mark=x,
%    mark options=solid,
%    mark size=2pt,
%    line width=1pt
%    ] table[x={N},  y = {Time}]{Data/ScaleFree/N/singleGreedy.txt}; ]
\addplot[
    color=blue,
    mark=o,
    mark options=solid,
    mark size=2pt,
    line width=1pt
    ] table[x={N},  y = {Time}]{Data/Random/N/allGreedy.txt}; ]
\addplot[
    color=green,
    mark=triangle,
    mark options=solid,
    mark size=2pt,
    line width=1pt
    ] table[x={N},  y = {Time}]{Data/Random/N/myopic.txt}; ]
\addplot[
    color=olive,
    mark=pentagon,
    mark options=solid,
    mark size=2pt,
    line width=1pt
    ] table[x={N},  y = {Time}]{Data/Random/N/bws.txt}; ]
\addplot[
    color=orange,
    mark=diamond,
    mark options=solid,
    mark size=2pt,
    line width=1pt
    ] table[x={N},  y = {Time}]{Data/Random/N/dynamic.txt}; ]
\legend{$\Psi$-Saturate, AllGreedy, Myopic, BWS, DP-RRP}
\end{axis}
\end{tikzpicture}
\end{subfigure}
\hspace*{1mm}
\begin{subfigure}{.45\linewidth}
\pgfplotsset{width=4.6cm,compat=1.18}
\begin{tikzpicture}
\begin{axis}[
    title style={yshift=-5 pt,},
    xlabel = \(n\),
    %ylabel = {Time (msec)},
    %xmode = log,
    %xmin=392, xmax=83772,
    %xtick={392,719,1185,1616,2224,3199,4845,9015,83772},
    legend pos=north west,
    legend style={nodes={scale=0.5, transform shape}},
    ymajorgrids=true,
    grid style=dashed,
    ,scaled y ticks={base 10:-2},
    %ticklabel style = {font=\tiny},
    %xlabel style={font=\tiny},
    %ylabel style={font=\tiny},
    %ylabel shift = -6 pt,
    %xlabel shift = -5 pt,
    %yticklabel shift = -1 pt,
    %xticklabel shift = -1 pt,
    legend cell align={left},
    xtick={2500, 5000, 7500, 10000, 12500},
    ytick={0,5000,10000},
    yticklabels={,$0.5$,$1$},
    ticklabel style = {font=\small},
    %xlabel style={font=\tiny},
    %ylabel style={font=\tiny},
    %ylabel shift = -10 pt,
    xlabel shift = -4 pt,
    every x tick scale label/.style={
    at={(1,0)},xshift=1pt,anchor=south west,inner sep=0pt
    },
    scaled x ticks=base 10:-3,
    scaled y ticks=base 10:-4
]
\addplot[
    color=red,
    mark=triangle,
    mark options=solid,
    mark size=2pt,
    line width=1pt
    ] table[x={N},  y = {Time}]{Data/ScaleFree/N/saturate.txt}; ]
%\addplot[
%    color=brown,
%    mark=x,
%    mark options=solid,
%    mark size=2pt,
%    line width=1pt
%    ] table[x={N},  y = {Time}]{Data/ScaleFree/N/singleGreedy.txt}; ]
\addplot[
    color=blue,
    mark=o,
    mark options=solid,
    mark size=2pt,
    line width=1pt
    ] table[x={N},  y = {Time}]{Data/ScaleFree/N/allGreedy.txt}; ]
\addplot[
    color=green,
    mark=triangle,
    mark options=solid,
    mark size=2pt,
    line width=1pt
    ] table[x={N},  y = {Time}]{Data/ScaleFree/N/myopic.txt}; ]
\addplot[
    color=orange,
    mark=diamond,
    mark options=solid,
    mark size=2pt,
    line width=1pt
    ] table[x={N},  y = {Time}]{Data/ScaleFree/N/dynamic.txt}; ]
\addplot[
    color=olive,
    mark=pentagon,
    mark options=solid,
    mark size=2pt,
    line width=1pt
    ] table[x={N},  y = {Time}]{Data/ScaleFree/N/bws.txt}; ]
%\legend{$\Psi$-Saturate, AllGreedy, Myopic, Dynamic, BWS}
%\legend{$\Psi$-Saturate, AllGreedy, Myopic, Dynamic, BWS}
\end{axis}
\end{tikzpicture}
\end{subfigure}
\vspace*{-3mm}
\\
\hspace*{0mm}
\begin{subfigure}{.45\linewidth}
\pgfplotsset{width=4.6cm,compat=1.18}
\begin{tikzpicture}
\begin{axis}[
    title style={yshift=-5 pt,},
    xlabel = \(L\),
    ylabel = {Time (msec)},
    %xmode = log,
    %xmin=392, xmax=83772,
    %xtick={392,719,1185,1616,2224,3199,4845,9015,83772},
    legend pos=north west,
    legend style={nodes={scale=0.5, transform shape}},
    ymajorgrids=true,
    grid style=dashed,
    ,scaled y ticks={base 10:-2},
    %ticklabel style = {font=\tiny},
    %xlabel style={font=\tiny},
    %ylabel style={font=\tiny},
    ylabel shift = 3 pt,
    %xlabel shift = -5 pt,
    %yticklabel shift = -1 pt,
    %xticklabel shift = -1 pt,
    legend cell align={left},
    xtick={0.1, 0.25 , 0.5, 0.75},
    xticklabels={,$0.25$,$0.5$,$0.75$},
    ticklabel style = {font=\small},
    %xlabel style={font=\tiny},
    %ylabel style={font=\tiny},
    %ylabel shift = -10 pt,
    xlabel shift = -4 pt,
    every x tick scale label/.style={
    at={(1,0)},xshift=1pt,anchor=south west,inner sep=0pt
    },
    %scaled x ticks=base 10:-3,
    scaled y ticks=base 10:-4
]
\addplot[
    color=red,
    mark=triangle,
    mark options=solid,
    mark size=2pt,
    line width=1pt
    ] table[x={budget},  y = {Time}]{Data/Random/Budget/saturate.txt}; ]
%\addplot[
%    color=brown,
%    mark=x,
%    mark options=solid,
%    mark size=2pt,
%    line width=1pt
%    ] table[x={N},  y = {Time}]{Data/ScaleFree/N/singleGreedy.txt}; ]
\addplot[
    color=blue,
    mark=o,
    mark options=solid,
    mark size=2pt,
    line width=1pt
    ] table[x={budget},  y = {Time}]{Data/Random/Budget/allGreedy.txt}; ]
\addplot[
    color=green,
    mark=triangle,
    mark options=solid,
    mark size=2pt,
    line width=1pt
    ] table[x={budget},  y = {Time}]{Data/Random/Budget/myopic.txt}; ]
\addplot[
    color=orange,
    mark=diamond,
    mark options=solid,
    mark size=2pt,
    line width=1pt
    ] table[x={budget},  y = {Time}]{Data/Random/Budget/dynamic.txt}; ]
\addplot[
    color=olive,
    mark=pentagon,
    mark options=solid,
    mark size=2pt,
    line width=1pt
    ] table[x={budget},  y = {Time}]{Data/Random/Budget/bws.txt}; ]
%\legend{Saturate, AllGreedy, Myopic, Dynamic, BWS}
\end{axis}
\end{tikzpicture}
\end{subfigure}
\hspace*{3mm}
\begin{subfigure}{.45\linewidth}
\pgfplotsset{width=4.6cm,compat=1.18}
\begin{tikzpicture}
\begin{axis}[
    title style={yshift=-5 pt,},
    xlabel = \(L\),
    %ylabel = {Time (msec)},
    %xmode = log,
    %xmin=392, xmax=83772,
    %xtick={392,719,1185,1616,2224,3199,4845,9015,83772},
    legend pos=north west,
    legend style={nodes={scale=0.5, transform shape}},
    ymajorgrids=true,
    grid style=dashed,
    ,scaled y ticks={base 10:-2},
    ticklabel style = {font=\small},
    %xlabel style={font=\tiny},
    %ylabel style={font=\tiny},
    %ylabel shift = -10 pt,
    xlabel shift = -4 pt,
    ticklabel style = {font=\small},
    %xlabel style={font=\tiny},
    %ylabel style={font=\tiny},
    %ylabel shift = -10 pt,
    xlabel shift = -4 pt,
    %ticklabel style = {font=\tiny},
    %xlabel style={font=\tiny},
    %ylabel style={font=\tiny},
    %ylabel shift = -6 pt,
    %xlabel shift = -5 pt,
    %yticklabel shift = -1 pt,
    %xticklabel shift = -1 pt,
    legend cell align={left},
    xtick={0.1, 0.25 , 0.5, 0.75},
    xticklabels={,$0.25$,$0.5$,$0.75$},
    every x tick scale label/.style={
    at={(1,0)},xshift=1pt,anchor=south west,inner sep=0pt
    },
    %scaled x ticks=base 10:-3,
    scaled y ticks=base 10:-4
]
\addplot[
    color=red,
    mark=triangle,
    mark options=solid,
    mark size=2pt,
    line width=1pt
    ] table[x={budget},  y = {Time}]{Data/ScaleFree/Budget/saturate.txt}; ]
%\addplot[
%    color=brown,
%    mark=x,
%    mark options=solid,
%    mark size=2pt,
%    line width=1pt
%    ] table[x={N},  y = {Time}]{Data/ScaleFree/N/singleGreedy.txt}; ]
\addplot[
    color=blue,
    mark=o,
    mark options=solid,
    mark size=2pt,
    line width=1pt
    ] table[x={budget},  y = {Time}]{Data/ScaleFree/Budget/allGreedy.txt}; ]
\addplot[
    color=green,
    mark=triangle,
    mark options=solid,
    mark size=2pt,
    line width=1pt
    ] table[x={budget},  y = {Time}]{Data/ScaleFree/Budget/myopic.txt}; ]
\addplot[
    color=orange,
    mark=diamond,
    mark options=solid,
    mark size=2pt,
    line width=1pt
    ] table[x={budget},  y = {Time}]{Data/ScaleFree/Budget/dynamic.txt}; ]
\addplot[
    color=olive,
    mark=pentagon,
    mark options=solid,
    mark size=2pt,
    line width=1pt
    ] table[x={budget},  y = {Time}]{Data/ScaleFree/Budget/bws.txt}; ]
%\legend{Saturate, AllGreedy, Myopic, Dynamic, BWS}
\end{axis}
\end{tikzpicture}
\end{subfigure}
%\vspace{-3mm}
\caption{Preprocessing and running time vs.~$n$, $L$ for \er (left) and \scfr (right) datasets.\label{fig:time}}
%\vspace{-4mm}
\end{figure}
\input{Plots/generated/performance}

\paragraph{Time efficiency.} Figure~\ref{fig:time} plots the average (over~$20$ graphs) preprocessing time vs. graph size and running time for all algorithms vs. graph size and budget. 
Notably, the precomputation takes time superlinear in graph size~$n$, as the time complexity needed for the power iteration is~$\mathcal{O}(K(n+m))$ for~$K$ steps maximum and~$m$ edges. 
%In addition, the 
The runtime of most algorithms grows linearly with graph size and budget, indicating their efficiency, %with the exception of 
except of 
DP-RRP, whose time complexity is at least quadratic in~$n$, $\Theta(|\Policies|Ln)$ while~$L = \Omega(n)$.

\paragraph{Reward placement robustness.} Figure~\ref{fig:performance} illustrates the algorithms' average performance (over~$20$ graphs) in the reward placement robustness objective as several parameters vary. On the \er data, $\Psi$-Saturate outperforms all heuristics vs. graph size~$n$ and in other measurements, while on \scfr data, DP-RRP performs best overall. With a single setting, i.e., when~$|\Policies|=1$, all heuristics find an almost optimal solution. However, as expected, performance decreases as the number of models~$|\Policies|$ grows, whence the adversary possesses a larger pool of models to select from. Similarly, as the number of steps~$K$ grows, the feasible agent movements expand, causing the optimal cumulative reward per setting to rise more than the worst-case reward in general, hence the robustness score falls. In contrast, the score of all algorithms rises with budget~$L$. Intuitively, a higher budget offers more flexibility to hedge against worst-case outcomes, hence better robustness scores. This effect is more evident on the \scfr dataset, which has fewer lucrative nodes of high in-degree. The growth of robustness scores vs.~$\langle d\rangle$ on \er data confirms the importance of in-degree. On the other hand, the growth of~$p_{\beta}$ results in \scfr networks with more skewed power-law in-degree distributions, on which robustness scores suffer.

\subsection{Real Data}

To further validate our observations, we create graphs using real-world movement data. We gathered movement records from Baidu Map, covering Xuanwu District in Nanjing\footnote{https://en.wikipedia.org/wiki/Nanjing} from July 2019 to September 2019; these records comprise sequential Points of Interest (POIs) with timestamps, allowing us to calculate the probability of transitioning between POIs based on the Markovian assumption. Using these probabilities, we construct graphs where nodes represent POIs and edges express transition probabilities. Each graph captures a~7-day period, resulting in a total of~13 graphs. The combined dataset features a total of~51\,943 different nodes. Out of practical considerations, we assign data-driven costs to POIs based on their visit frequency and a fixed value: $c[x] = \lfloor\mathsf{frequency}(x)/25 + 50\rfloor$. The initial and steps probabilities follow the same default setup as the synthetic datasets.

\paragraph{Time and Performance.} Preprecessing the Xuanwu dataset where~$n = 51\,943$ and $|\Policies| = 13$ takes~118 seconds. Figure~\ref{fig:Xuanwu} shows how the running time and robustness scores vary as the budget grows. DP-RRP is the most time-consuming, followed by $\Psi$-Saturate, while BWS emerges as the most time-efficient solution. Interestingly, the robustness score does not follow a clear upward trend vs. budget with these real-world data; after all, the objective of Equation~\eqref{eq:RRP} is not a monotonic function of budget; while a higher budget allows for more flexibility in allocating resources, it also allows the optimal reward to grow correspondingly. Nevertheless, DP-RRP consistently outperforms all algorithms across budget values, corroborating its capacity to uncover high quality solutions even in hard scenarios, even while the performance of $\Psi$-Saturate and other heuristics fluctuates.

\pgfplotsset{width=5cm,compat=1.18}
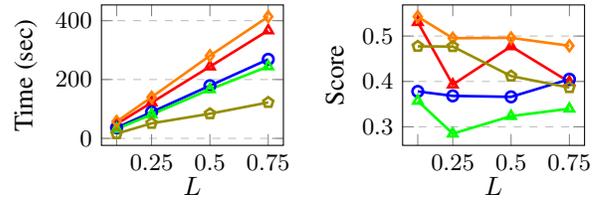
\begin{figure}[!t]
\centering
\begin{subfigure}{.45\linewidth}
\pgfplotsset{width=4.0cm,compat=1.18}
\begin{tikzpicture}
\begin{axis}[
    title style={yshift=-5 pt,},
    xlabel = \(L\),
    ylabel = {Time (sec)},
    %xmode = log,
    %xmin=392, xmax=83772,
    %xtick={392,719,1185,1616,2224,3199,4845,9015,83772},
    legend pos=north west,
    legend style={nodes={scale=0.4, transform shape}},
    ymajorgrids=true,
    grid style=dashed,
    ticklabel style = {font=\small},
    %xlabel style={font=\tiny},
    %ylabel style={font=\tiny},
    %ylabel shift = -10 pt,
    xlabel shift = -4 pt,
    %,scaled y ticks={base 10:-2},
    %ticklabel style = {font=\tiny},
    %xlabel style={font=\tiny},
    %ylabel style={font=\tiny},
    %ylabel shift = -6 pt,
    %xlabel shift = -5 pt,
    %yticklabel shift = -1 pt,
    %xticklabel shift = -1 pt,
    legend cell align={left},
    every x tick scale label/.style={
    at={(1,0)},xshift=1pt,anchor=south west,inner sep=0pt
    },
    xtick={0.1, 0.25 , 0.5, 0.75},
    xticklabels={,$0.25$,$0.5$,$0.75$},
    %scaled x ticks=base 10:-3,
    %scaled y ticks=base 10:-5
]
\addplot[
    color=red,
    mark=triangle,
    mark options=solid,
    mark size=2pt,
    line width=1pt
    ] table[x={budget},  y expr=\thisrow{Time}*0.001]{Data/xw/Budget/saturate.txt}; ]
%\addplot[
%    color=brown,
%    mark=x,
%    mark options=solid,
%    mark size=2pt,
%    line width=1pt
%    ] table[x={N},  y = {Time}]{Data/ScaleFree/N/singleGreedy.txt}; ]
\addplot[
    color=blue,
    mark=o,
    mark options=solid,
    mark size=2pt,
    line width=1pt
    ] table[x={budget},  y expr=\thisrow{Time}*0.001]{Data/xw/Budget/allGreedy.txt}; ]
\addplot[
    color=green,
    mark=triangle,
    mark options=solid,
    mark size=2pt,
    line width=1pt
    ] table[x={budget},  y expr=\thisrow{Time}*0.001]{Data/xw/Budget/myopic.txt}; ]
\addplot[
    color=orange,
    mark=diamond,
    mark options=solid,
    mark size=2pt,
    line width=1pt
    ] table[x={budget}, y expr=\thisrow{Time}*0.001]{Data/xw/Budget/dynamic.txt}; ]
\addplot[
    color=olive,
    mark=pentagon,
    mark options=solid,
    mark size=2pt,
    line width=1pt
    ] table[x={budget},  y expr=\thisrow{Time}*0.001]{Data/xw/Budget/bws.txt}; ]
%\legend{$\Psi$-Saturate, AllGreedy, Myopic, Dynamic, BWS}
\end{axis}
\end{tikzpicture}
\end{subfigure}
\hspace*{1mm}
\begin{subfigure}{.45\linewidth}
\pgfplotsset{width=4.0cm,compat=1.18}
\begin{tikzpicture}
\begin{axis}[
    title style={yshift=-5 pt,},
    xlabel = \(L\),
    ylabel = {Score},
    %xmode = log,
    %xmin=392, xmax=83772,
    %xtick={392,719,1185,1616,2224,3199,4845,9015,83772},
    legend pos=north west,
    legend style={nodes={scale=0.7, transform shape}},
    ymajorgrids=true,
    grid style=dashed,
    ticklabel style = {font=\small},
    %xlabel style={font=\tiny},
    %ylabel style={font=\tiny},
    %ylabel shift = -10 pt,
    xlabel shift = -4 pt,
    %,scaled y ticks={base 10:-2},
    %ticklabel style = {font=\tiny},
    %xlabel style={font=\tiny},
    %ylabel style={font=\tiny},
    %ylabel shift = -6 pt,
    %xlabel shift = -5 pt,
    %yticklabel shift = -1 pt,
    %xticklabel shift = -1 pt,
    legend cell align={left},
    xtick={0.1, 0.25 , 0.5, 0.75},
    xticklabels={,$0.25$,$0.5$,$0.75$},
    every x tick scale label/.style={
    at={(1,0)},xshift=1pt,anchor=south west,inner sep=0pt
    },
    %scaled y ticks=base 10:1
    %scaled x ticks=base 10:-3,
    %scaled y ticks=base 10:-3
]
\addplot[
    color=red,
    mark=triangle,
    mark options=solid,
    mark size=2pt,
    line width=1pt
    ] table[x={budget},  y = {Score}]{Data/xw/Budget/saturate.txt}; ]
%\addplot[
%    color=brown,
%    mark=x,
%    mark options=solid,
%    mark size=2pt,
%    line width=1pt
%    ] table[x={N},  y = {Time}]{Data/ScaleFree/N/singleGreedy.txt}; ]
\addplot[
    color=blue,
    mark=o,
    mark options=solid,
    mark size=2pt,
    line width=1pt
    ] table[x={budget},  y = {Score}]{Data/xw/Budget/allGreedy.txt}; ]
\addplot[
    color=green,
    mark=triangle,
    mark options=solid,
    mark size=2pt,
    line width=1pt
    ] table[x={budget},  y = {Score}]{Data/xw/Budget/myopic.txt}; ]
\addplot[
    color=orange,
    mark=diamond,
    mark options=solid,
    mark size=2pt,
    line width=1pt
    ] table[x={budget},  y = {Score}]{Data/xw/Budget/dynamic.txt}; ]
\addplot[
    color=olive,
    mark=pentagon,
    mark options=solid,
    mark size=2pt,
    line width=1pt
    ] table[x={budget},  y = {Score}]{Data/xw/Budget/bws.txt}; ]
%\legend{Saturate, AllGreedy, Myopic, Dynamic, BWS}
\end{axis}
\end{tikzpicture}
\end{subfigure}
\caption{Time and robustness score vs.~$L$ on the Xuanwu dataset.}\label{fig:Xuanwu}
\end{figure}
%-------------------------------------------------------------

%Conclusions
%-------------------------------------------------------------
\section{Conclusions}

We introduced the \NPhard problem of \problem (RRP). Assuming an agent is moving on an unknown Markov Mobility Model (MMM), sampled by a set of $\Policies$ candidates, RRP calls to select reward states within a budget that maximize the worst-case ratio of the expected reward (agent visits) over the optimal one. Having shown that RRP is strongly inapproximable, we propose $\Psi$-Saturate, an algorithm that achieves an~$\epsilon$-additive approximation by exceeding the budget constraint by a factor of~$\mathcal{O}(\ln |\Pi|/\epsilon)$. We also developed heuristics, most saliently one based on a dynamic programming algorithm. Our experimental analysis on both synthetic and real-world data indicates the effectiveness of $\Psi$-Saturate and the dynamic-programming-based solution. In the future, we aim to examine the robust configuration of agent-based content features~\cite{ivanov17} under a set of adversarial mobility models.
%\clearpage
%-------------------------------------------------------------

%Acknowledgments
%-------------------------------------------------------------
\section*{Acknowledgments}
Work supported by grants from DFF (P.P. and P.K., 9041-00382B) and Villum Fonden (A.P., VIL42117).
%-------------------------------------------------------------

%References
%-------------------------------------------------------------
\bibliographystyle{named}
\bibliography{ref}
\clearpage
%-------------------------------------------------------------

%Appendix

%Appendix
%-------------------------------------------------------------
\appendix
\section{Appendix}\label{sec:app}

\lemrpminmax*
\begin{proof}
To make the reduction, we first calculate the optimal reward placement for each model $\policy\in\Policies$, separately. 
This can be done in $\mathcal{O}(|\Policies|Ln)$, as for a given model $\policy$, the optimal reward placement $\RewardStatesPolicy = \arg_{\RewardStates} \max \ExpReward(\RewardStates|\policy)$ within budget~$L$ can be found in~$\mathcal{O}(Ln)$ due to Lemma~\ref{lem:rp_optimal_linear}. 
Then, for each state $\state \in \States$ we set a unique item $u \in U$ and each model $\policy_i \in \Policies$ corresponds to unique scenario $x_i\in X$. Each item has cost~$c(u) = c[\state]$ with value $F_{x_i}(u)=\frac{\ExpReward(\{\state\}|\policy_i)}{\ExpReward(\pki{\RewardStatesPolicyi}|\policy_i)}$. 
Notably, the optimal solution for \mnk is also optimal for \problemshort, as by using equation $\sum_{u\in V}F_{x_i}(u) = \sum_{\state \in \RewardStates}{\ExpReward( \{\state\}|\policy_i)} = \ExpReward(\pki{\RewardStates}|\policy_i)$ we have:
$\max\limits_V\min\limits_{x_i\in X}\sum\limits_{u\in V}F(u) \,\,=\max\limits_{\RewardStates} \min \limits_{\policy_i \in \Policies} \frac{\ExpReward(\RewardStates|\policy_i )}{\ExpReward(\RewardStatesPolicyi|\policy_i)}$
\end{proof}

\lemheur*
\begin{proof}
Lemma~\ref{lem:bicriteria_apx} states that for any constant $\epsilon>0$ with $\beta = 1 + \ln\frac{3|\Policies|}{\epsilon}$, $\Psi$-Saturate  achieves an $\epsilon$-additive approximation of the optimal solution by exceeding the budget constraint from $L$ to $\beta L$.
In contrast, heuristics cannot guarantee any bicriteria approximation under the same budget increase.
Consider the \problemshort~reduction of a \hitset instance as in Theorem~\ref{th:inapprox}. Under this scenario, the min ratio objective in Equation~\eqref{eq:RRP} is larger than zero if and only if the reward placement forms a \Petros{hitting set.} 
The All Greedy solution contains elements from a single set; thus, even by exceeding the budget constraint, the objective might equal zero, as it does not consider forming a hitting set.
Similarly, BWS and Myopic iteratively select the state that offers maximal marginal gain. In the worst case, both may be ineffective, as the marginal gain of each state can be zero even after exceeding the budget unless it forms a hitting set. 
Finally, even by exceeding the budget constraint to~$\beta L$, with~$\beta = 1+\ln\frac{3|\Policies|}{\epsilon}$, DP-RRP may perform arbitrarily bad. In the worst case, entry $M[n, \lfloor\beta L\rfloor]$ may contain a tuple with at least one zero, resulting in an objective value of zero.
This is because Equation~\eqref{eq:Mrec} iteratively selects the maximum of two tuples without guaranteeing optimality.
\end{proof}
%-------------------------------------------------------------
\end{document}